%% file: main.tex
\documentclass{article}
\usepackage{ijcai26}

\usepackage{times}
\usepackage{soul}
\usepackage{url}
\usepackage[hidelinks]{hyperref}
\usepackage[utf8]{inputenc}
\usepackage[small]{caption}
\usepackage{graphicx}
\usepackage{amsmath}
\usepackage{amsthm}
\usepackage{booktabs}
\usepackage{algorithm}
\usepackage{algorithmic}
\usepackage[switch]{lineno}
\usepackage{amsfonts,bm,mathtools,amssymb}
\usepackage{stmaryrd}
\usepackage[T1]{fontenc}
\usepackage{todonotes}
\usepackage{xspace}

\input{macros}

\newtheorem{theorem}{Theorem}

\title{The Complexity of Verifying Feedforward Neural Networks in Quantised Settings}

\author{
Eric Alsmann$^1$
\and
Martin Lange$^1$\and
Marco Sälzer$^{2}$ \\
\affiliations
$^1$University of Kassel\\
$^2$RPTU University Kaiserslautern-Landau\\
\emails
\{eric.alsmann, martin.lange\}@uni-kassel.de,
marco.saelzer@rptu.de
}

\newtheorem{lemma}[theorem]{Lemma}

\newtheorem{proposition}[theorem]{Proposition}

\theoremstyle{definition}

\begin{document}

\maketitle

\begin{abstract}
    We investigate the computational complexity of neural network verification in quantised settings. 
    We distinguish three classes of Feedforward Neural Networks (FNNs): rational FNNs with exact rational 
    weights, quantised FNNs whose weights come from a finite-width arithmetic, 
    and dynamically quantised FNNs in which rational networks are evaluated with respect to a 
    given finite-width arithmetic. We consider two types of specifications used in the literature. 
    Linear programming (LP) specifications are conjunctions of linear constraints, while bit-vector (BV) 
    specifications allow reasoning at the bit level and can express non-linear constraints. 
    Our results give a complexity landscape of these verification problems. 
    For quantised FNNs with fixed arithmetic precision, we show that verification under both LP and BV specifications remains NP-complete, 
    matching the complexity of the rational case. For dynamically quantised FNNs with BV specifications, we establish upper bounds, 
    complementing a previously known PSPACE-hardness result.
\end{abstract}

\section{Introduction}
\label{sec:intro}
\input{sections/introduction.tex}

\section{Preliminaries}
\label{sec:prelims}
\input{sections/prelims.tex}

\section{Verifying Quantised FNN: Overview}
\label{sec:problems}
\input{sections/overview_problems.tex}

\section{Verifying Quantised FNN in Relation to Linear Constraints}
\label{sec:lpproblems}
\input{sections/lp_problems.tex}

\section{Verifying Quantised FNN in Relation to BV Constraints}
\label{sec:bvproblems}
\input{sections/bv_problems.tex}

\section{Outlook}
\label{sec:outlook}
\input{sections/outlook.tex}





\bibliographystyle{named}
\bibliography{refs}

\pagebreak
\appendix
\section{Omitted proofs from Section~\ref{sec:bvproblems}}
\label{sec:appendix}
\input{sections/appendix.tex}

\end{document}

%% file: macros.tex
\newcommand*{\rats}{\mathbb{Q}}
\newcommand*{\ints}{\mathbb{Z}}
\newcommand*{\nats}{\mathbb{N}}

\newcommand{\bs}[1]{\boldsymbol{#1}}
\newcommand*{\size}[2]{\ensuremath{|#2|_{#1}}}

\newcommand*{\fixarith}{\mathbb{F}^{\mathrm{ix}}}
\newcommand*{\floatarith}{\mathbb{F}^{\mathrm{loat}}}
\newcommand*{\arith}{\mathbb{F}}

\newcommand*{\relu}{\mathrm{relu}}

\newcommand*{\NP}{\textsc{NP}\xspace}

\newcommand*{\PSPACE}{\textsc{PSpace}\xspace}
\newcommand*{\NEXPTIME}{\textsc{NExpTime}\xspace}

\newcommand*\BitAnd{\mathbin{\&}}
\newcommand*\BitOr{\mathbin{|}}

\newcommand*\BitNeg{\ensuremath{\mathord{\sim}}}

%% file: sections/introduction.tex

Feedforward Neural Networks (FNN) are ubiquitous in deep learning, not only as a stand-alone model for classifying
inputs or transforming input values into outputs, but also as components in other machine learning models, for instance
transformers, cf.\ \cite{VaswaniSPUJGKP17}. The ever-growing quest for machine-learning methods to solve various 
computational tasks, in particular for safety-critical applications, comes with a need for methods that formally 
verify properties like safety or robustness of such models, including FNN.

In this paper, we are concerned with verification tasks of the following form: given an FNN $N$ and 
a safety property $\varphi$, decide whether $\varphi$ holds for $N$ or not.
To distinguish this from other forms of verification, such as testing $N$ on
a sample set of inputs or adjusting the training procedure to meet certain safety criteria, 
we refer to this as \emph{static (FNN) verification}. 

In its most rigorous form, static verification tasks are stated as formal decision
problems. In one of the most prominent variants, addressed by a wide-range of practical solvers, cf.\ 
\cite{VNNComp2024,BrixMBJL23}, or analysed formally, cf.\ \cite{DBLP:journals/fuin/SalzerL22,Wurm23,Henzinger_Lechner_Žikelić_2021}, 
static verification is concerned with the question of whether the pairs of inputs and corresponding 
outputs of a given FNN satisfy certain linear constraints.
Problems of this kind are often referred to as \emph{reachability} problems, cf.\ \cite{HuangKRSSTWY20}, and 
other formal problems like robustness can usually be reduced to reachability.

The reachability problem for FNN was claimed (see  
\cite{DBLP:conf/cav/KatzBDJK17,DBLP:journals/fmsd/KatzBDJK22,DBLP:conf/ijcai/RuanHK18}) and shown (see
\cite{10.1007/978-3-030-89716-1_10,DBLP:journals/fuin/SalzerL22}) to be \NP-complete. In this setting, FNN
are considered to be mappings on real-valued vector spaces. In practice, the arithmetic operations internal
to the FNN are of course executed not in arithmetic over $\mathbb{R}$ which would be infeasible on any 
bounded-memory computer, but in some fixed-width arithmetic using 64bit floating-point or fixed-point
numbers for instance. This shift from real-valued inputs and outputs to such discrete values is also known 
as \emph{quantisation}, and it comes not only with a loss of precision but also with artefacts like rounding
or over-/underflows. One may therefore argue that (\textsc{i}) \NP-completeness of the reachability problem for 
FNN in arithmetic over $\mathbb{R}$ does not necessarily predict the computational complexity of
formal verification for FNN over a quantised arithmetic, and (\textsc{ii}) that the quantised setting
is more interesting for practical purposes anyway.  

We investigate the computational complexity of formal verification, i.e.\ reachability,
for quantised FNN. Quantisation does not make formal verification easier, i.e.\ \NP-hardness remains a lower bound.
This is not really surprising since computational hardness is often tightly linked to discretisation, see
for instance the \NP-completeness of Mixed-Integer Linear Programming versus the polynomial-time
solvability of $\mathbb{R}$-valued linear programs.
\cite{Henzinger_Lechner_Žikelić_2021} recently studied the verification problem for quantised 
FNN and provided a result that suggests higher complexity, namely \PSPACE instead of \NP. There seems to be
some widespread misunderstanding about the actual results, perhaps fuelled by the fact that the paper is
seen to be hard to read. First of all, \cite{Henzinger_Lechner_Žikelić_2021} only provides a lower bound of 
\PSPACE-hardness, not \PSPACE-completeness. Second, they consider a richer specification language for safety 
properties: while \NP-completeness holds for safety properties that are convex sets of vectors, i.e.\ 
specifications in the form of linear programs, \cite{Henzinger_Lechner_Žikelić_2021} uses bit-vector logics 
which allow non-convex sets to be specified. Thus, the increase in complexity from \NP to \PSPACE is not due to
quantisation, as that paper may suggest, but due to the use of a richer specification language for safety
properties.

In this paper, we aim to provide an overall picture of the computational complexity of formal verification
for quantised FNN. We classify the complexities depending on (\textsc{i}) the involved specification formalisms 
for safety properties, as they are commonly used in the literature, (\textsc{ii}) the exact nature of the
fixed bit-width arithmetic that determines how the FNN's internal computations are performed on common
hardware architectures and, finally, (\textsc{iii}) the dependence on the parameter determining the bit-width
of such numbers, in particular whether this is given in unary or in binary.
The motivation for studying the effect that this parameter has, also comes from practice: suppose we
import a trained FNN to run on our hardware. The FNN may come with advice saying that arithmetic operations
should be carried out with a certain bit-width. In this case, that width would normally be given as a decimal
number, i.e.\ the actual width is exponential in the representation size of that parameter. On the other hand,
when verifying a quantised FNN on the hardware where it already resides, the bit-width rather corresponds to
the size of actual CPU registers, and therefore a unary encoding of this length may be a more realistic measure
of the space needed to represent such quantised values.

The paper is organised as follows. In Section~\ref{sec:prelims}, we provide thorough definitions of all
fundamentals necessary for the findings presented in this paper. In Section~\ref{sec:problems}, we formally 
present reachability problems of FNN as decision problems. In Section~\ref{sec:lpproblems}, we examine the 
complexity of reachability in different quantisation settings where specifications are given by linear constraints. 
In Section~\ref{sec:bvproblems}, we study the complexity of verifying quantised FNN in the context of bit-vector 
logic, complementing the work of \cite{Henzinger_Lechner_Žikelić_2021}. Finally, in Section~\ref{sec:outlook}, 
we briefly discuss our findings and outline prospective next steps.

%% file: sections/prelims.tex

\paragraph{Numbers, vectors and matrices.} We exclusively consider numbers from  $\rats$. 
Typically, we denote vectors using small, bold symbols like $\bs{x}$, $\bs{y}$ or $\bs{z}$ and matrices using big, 
bold symbols like $\bs{A}$, $\bs{B}$ or $\bs{C}$. We denote by $\bs{x}_{i}$ the $i$-th entry of vector 
$\bs{x}$ and by $\bs{A}_{i,j}$ the entry in the $i$-th row and $j$-th column of matrix $\bs{A}$. Let 
$x \in \rats$ with $p \in \ints$ and $q \in \nats^{> 0}$ such that $x = \frac{p}{q}$. 
If not stated otherwise, we define the \emph{size of $x$}, denoted by $\size{}{x}$, as 
$\size{}{x} = \size{}{p} + \size{}{q}$
where $\size{}{p} = (\lfloor \log_2(p) + 1 \rfloor) + 1$ and $\size{}{q} = \lfloor \log_2(q) + 1 \rfloor$. Analogously, 
we extend the definition of $\size{}{\cdot}$ to vectors $\bs{x} \in \rats^m$ and matrices $\bs{A} \in \rats^{m\times n}$ by 
$\size{}{\bs{x}} = \sum_{i=1}^m \size{}{\bs{x}_i}$ and $\size{}{\bs{A}} = \sum_{i=1}^m \sum_{j=1}^n\size{}{\bs{A}_{i,j}}$.

\paragraph{Finite-width arithmetics.} We consider two kinds of finite arithmetics: fixed- and floating-point arithmetics. 
A \emph{fixed-point arithmetic $\fixarith_{b,f}$} is a tuple $(\rats^\mathrm{fix}_{b,f}, o, r)$, where 
$\rats_{b,f} = \{\frac{n}{2^f} \mid -2^{b-1} \leq n < 2^{b-1}, n \in \ints\}$ is the sets of rationals representable using 
$b$ bits with $f$ bits for representing the fractional parts, $o \colon \rats \to \rats_{b,f}$ is a map called \emph{overflow mode}
and $r \colon \rats \rightarrow \rats_{b,f}$ is a map called \emph{rounding mode}. We associate with $\fixarith_{b,f}$
the typical arithmetic operations, carried out as $\circ_{\fixarith_{b,f}} := o(r(x \circ y))$ for all 
$\circ \in \{+,\cdot,\div\}, x,y \in \rats$, and binary comparisons, carried out as $\sim_{\fixarith_{b,f}} := o(x) \sim o(y)$
for all $\sim \, \in \{<, \leq, =\}, x,y \in \rats$.
Similarly, a \emph{floating-point arithmetic 
$\floatarith_{p,e}$} is a tuple $(\rats^\mathrm{float}_{p,e}, r)$ where
\begin{align*}
    \rats^\mathrm{float}_{p,e} = \{&(-1)^s\cdot (1+\frac{k}{2^p})\cdot 2^E \mid s \in \{0,1\}, E,k \in \nats, \\
    &-(2^{e-1}-2) \leq E \leq 2^{e-1}-1, 0 \leq k \leq 2^p-1 \}
\end{align*}
is the set of rationals representable with $p$-bit significant and $e$-bit exponent width, and $r \colon \rats \to \rats^\mathrm{float}_{p,e}$ 
is a rounding mode. We associate with $\floatarith_{p,e}$
the typical arithmetic operations, carried out as $\circ_{\floatarith_{p,e}} := r(x \circ y)$ for all 
$\circ \in \{+,\cdot,\div\}, x,y \in \rats$, and  binary comparisons, carried out as $\sim_{\floatarith_{p,e}} := r(x) \sim r(y)$
for all ${\sim} \, \in \{<, \leq, =\}$ and $x,y \in \rats$.
To keep things concise, we refer to both of these as 
\emph{finite-width arithmetics $\arith_{n_1,n_2}$} and represent them as tuples $\arith = (n_1,n_2,r,o)$, where $n_1,n_2$ are the 
total bit number and fractional width (fixed-point), or the significant and exponent width (floating-point); $r$ is the rounding mode 
and $o$ the overflow mode (in the case of floating-point, we simply assume $o=r$ w.l.o.g.). We use $\arith_{n_1,n_2}$ 
to refer to the full arithmetic as well as the set of numbers representable in $\arith_{n_1,n_2}$.

\paragraph{Feedforward neural networks.} An FNN $N$, of some class of FNN $\mathcal{N}$, is a tuple $(\bs{A}^{(i)}, \bs{b}^{(i)})_{i=1}^k$ of matrices 
$\bs{A}^{(i)} \in \rats^{m_i \times n_i}$ and vectors $\bs{b}^{(i)} \in \rats^{m_i}$. We call $k$ the \emph{number of layers (or depth) of $N$}, 
$n_1$ the \emph{input dimension of $N$}, $n_k$ the \emph{output dimension of $N$}, $m_i$ the \emph{size of layer $i$} 
for all $i \leq k$, and generally assume that $m_i = n_{i+1}$ for all $i \leq k-1$. This concise representation 
of FNN from $\mathcal{N}$ assumes unambiguity in the definition of $\mathcal{N}$ regarding which and how activation functions 
are used by $N \in \mathcal{N}$. We generally assume that this is $\relu(x) = \max(0,x)$, the 
prominent \emph{ReLU activation} function, across all layers.\footnote{As done in other works such as 
\cite{DBLP:journals/fuin/SalzerL22}, our results can be extended to FNN using general
\emph{piecewise linear} activation functions.}

Regarding different assumptions about the underlying arithmetic, we consider two distinct types of FNN for verification:

\textbf{Rational FNN:} A \emph{rational FNN} $N = (\bs{A}^{(i)}, \bs{b}^{(i)})_{i=1}^k$ has weights and biases from $\rats$. It computes a function 
$N \colon \rats^{n_1} \to \rats^{n_k}$ by the standard forward propagation with $N(\bs{x}) =$
\begin{displaymath}
     \sigma^{(k)}(\bs{A}^{(k)} \cdot (\sigma^{(k-1)}( \dotsb \sigma^{(1)}(\bs{A}^{(1)} \cdot \bs{x} + \bs{b}^{(1)}) \dotsb )) + \bs{b}^{(k)})\ .
\end{displaymath}
For rational FNNs, we define $\size{}{N} = \sum_{i=1}^{k} \size{}{\bs{A}^{(i)}} + \size{}{\bs{b}^{(i)}}$.

\textbf{Quantised FNN:} Given some finite-width arithmetic $\arith = (n_1,n_2,r,o)$, a \emph{quantised FNN} has weights and biases from $\arith_{n_1,n_2}$. Such an FNN $N|_\arith$ computes a function $N|_\arith \colon \arith^{n_1} \to \arith^{n_k}$ by $N|_\arith(\bs{x}) =$ 
\begin{displaymath}
    \sigma^{(k)}_\arith(\bs{A}^{(k)} \cdot_\arith (\sigma^{(k-1)}_\arith( \dotsb \sigma^{(1)}_\arith(\bs{A}^{(1)} \cdot_\arith \bs{x} +_\arith \bs{b}^{(1)}) \dotsb )) +_\arith \bs{b}^{(k)})
\end{displaymath}
where $\sigma^{(i)}(\cdot)$ denotes the entrywise application of activation function $\sigma^{(i)} \colon \rats \to \rats$. We assume that 
the definition of $\arith$ determines how $\sigma^{(i)}$ is carried out in $\arith$. We associate with $N|_\arith$ the tuple 
$(\bs{A}^{(i)}_\arith, \bs{b}^{(i)}_\arith)_{i=1}^k$ where all weights and biases are represented in $\arith_{n_1,n_2}$. 
For quantised FNN, we define $\size{}{N|_\arith} = \sum_{i=1}^{k} \size{\arith}{\bs{A}^{(i)}} + \size{\arith}{\bs{b}^{(i)}}$, where 
$\size{\arith}{\cdot}$ denotes the representation size as specified by $\arith$.

\paragraph{Linear Programs.} A \emph{linear program (LP)} is a system of linear constraints of the form $L =$
    $\bigwedge_{i=1}^{m} \big( \sum_{j=1}^{n} a_{i,j} \cdot x_j \sim_i b_i \big)$
where $a_{i,j}, b_i \in \rats$, $x_j$ are variables, and ${\sim_i} \in \{<, \leq, =, \geq, >\}$ for each $i$. 
The \emph{feasibility problem} for such linear programs is: given an LP, does there exist an assignment of values 
to the variables $x_1, \ldots, x_n$, that satisfies all constraints? We can encode neural network verification problems 
as feasibility problems over linear constraints, where the constraints capture both the network's behaviour and the 
property specification. We define $\size{}{L} = \sum_{i=1}^{m} \size{}{b_i} + \sum_{j=1}^{n} \size{}{a_{i,j}}$.

A fundamental result in linear programming theory establishes a polynomial bound on the size of solutions:

\begin{theorem}[\cite{korteCombinatorialOptimizationTheory2018}]
\label{thm:lp-solution-size}
Let $L$ be a linear program with input size $\size{}{L}$.
If $L$ has a solution, then there exists a solution $(x_1, \ldots, x_n)$ such that $\sum_{j=1}^{n} \size{}{x_j} \leq \text{poly}(\size{}{L})$.
\end{theorem}

This polynomial bound is crucial for our analysis as it ensures that solutions to linear programs used in neural network verification have representations whose size is polynomially bounded by the input representation size.

\paragraph{Bit-Vector Logics (BV).}
We consider fixed-size bit vector logics as introduced in \cite{DBLP:journals/mst/KovasznaiFB16}. Let $\ell\in\mathbb{N}^+$ be a fixed bit-width and $\mathcal{V}$ be a finite set of variables. A BV formula $\varphi$ is inductively defined as
\begin{align*}
    \varphi &:= (t_1 = t_2) \mid \neg \varphi \mid \varphi \land \varphi  &
    t &:= x \mid c \mid \BitNeg t \mid t \BitAnd t \mid (t \BitOr t) 
\end{align*}
where $x\in\mathcal{V}$ and $0 \leq c \leq 2^\ell -1$.

A BV formula $\varphi$ is evaluated over bit-vectors of fixed size $\ell$. An assignment is a function $\vartheta: \mathcal{V} \rightarrow \{0,1\}^{\ell}$ that maps each variable to a bit-vector of length $\ell$. The semantics of formulas is defined inductively as
\begin{align*}
    \vartheta \models (t_1=t_2) &\Leftrightarrow \llbracket t_1 \rrbracket^{\vartheta} =  \llbracket t_2 \rrbracket^{\vartheta}\\
    \vartheta \models \varphi_1 \land \varphi_2 &\Leftrightarrow \vartheta \models \varphi_1 \text{ and } \vartheta \models \varphi_2 \\
    \vartheta \models \neg \varphi &\Leftrightarrow \vartheta \not\models \varphi
\end{align*}
The interpretation of terms under an assignment $\vartheta$ is defined as
\begin{align*}
    \llbracket x \rrbracket^{\vartheta} &= \vartheta(x) &
    \llbracket c \rrbracket^{\vartheta} &= \text{enc}_\ell(c) \\
    \llbracket \BitNeg t \rrbracket^{\vartheta} &= \BitNeg \llbracket t \rrbracket^{\vartheta} &
    \llbracket t_1 \odot t_2 \rrbracket^{\vartheta} &= \llbracket t_1 \rrbracket^{\vartheta} \odot \llbracket t_2 \rrbracket^{\vartheta}
\end{align*}
for $\odot\in\{\BitAnd, \BitOr\}$, where `$\BitAnd$', `$\BitOr$' and `$\BitNeg$' denote the standard bitwise AND, OR, and NOT operations
and $\text{enc}_\ell(c)$ is the $\ell$-bit binary representation of $c$.

A BV formula $\varphi$ is \emph{satisfiable over bit-vectors of length $\ell$} if there exists an assignment $\vartheta$ such that $\vartheta \models \varphi$.

\begin{proposition}[\cite{DBLP:journals/mst/KovasznaiFB16}]
    The satisfiability problem for BV formulas with bit length and constants given in binary is \textsc{NP}-complete.
\end{proposition}

The \emph{model checking} problem for BV is: given a BV formula $\varphi$ with bit-width $b$ and an assignment $\vartheta$, 
determine whether $\vartheta \models \varphi$.

\begin{proposition}
    \label{prop:bv_nlogspace}
    The model-checking problem for BV formulas can be decided in polynomial time.
\end{proposition}
\begin{proof}
    This is trivial, as given an assignment $\vartheta$, one can evaluate $\varphi$ bottom-up by traversing the syntax tree. Each term can be evaluated in time $\mathcal{O}(b)$, thus the whole procedure can be done in time $\mathcal{O}(b\cdot |\varphi|)$.
\end{proof}

%% file: sections/overview_problems.tex
\newtheorem{problem}{Definition}

We consider five distinct reachability problems in this paper, differentiated
by the type of arithmetic presumed for numerical operations and two
separate formalisms utilised to specify valid inputs and outputs in the literature.
We start with the most standard, yet somewhat idealised setting.

\paragraph*{$\textsc{Reach}_\mathbb{Q}(\text{LP})$:}
given rational FNN $N$, and two rational linear programs 
$L_1,L_2$, decide whether there is input $\bs{x}$ of $N$ that satisfies $L_1$ such that the output $N(\bs{x})$ satisfies 
$L_2$.

$\textsc{Reach}_\mathbb{Q}(\text{LP})$ assumes the computations of $N$, 
including the representation of its numerical parameters, to be performed in full-precision, rational arithmetic. 
However, this assumption
is unrealistic for concrete computer architectures, which are limited to specific finite-width arithmetics, such as IEEE 754
floating-point arithmetic using 32 bits to represent numbers in binary. Neglecting these real-world restrictions, as done by various
solvers addressing $\textsc{Reach}_\mathbb{Q}(\text{LP})$ or similar problems relying on idealised assumptions, has
severe consequences as shown in \cite{DBLP:conf/sas/JiaR21}. Thus, it is crucial to account for the quantisation of FNN
due to some finite-width arithmetic during verification. Let $\mathcal{F} = \{\arith_{n_1,n_2} \mid n_1,n_2 \in \nats\}$ denote the 
set of all finite-width arithmetics as defined in Section~\ref{sec:prelims}.
We assume w.l.o.g.\ that $\mathcal{F}$ only contains well-defined arithmetics, excluding, for instance, 
fixed-point arithmetics $\fixarith_{b,f}$ with $f > b$.

\paragraph{$\textsc{Reach}_\mathcal{F}(\text{LP})$:} 
given a quantised FNN $N|_{\arith}$ with input dimensionality $d$, and two quantised linear programs $L_1|_{\arith},L_2|_{\arith}$ 
for some $\arith \in \mathcal{F}$, decide whether there is $\bs{x} \in \arith^d$ that satisfies $L_1|_{\arith}$
such that $N|_{\arith}(\bs{x})$ satisfies $L_2|_{\arith}$.

In the context of quantised FNN, \cite{Henzinger_Lechner_Žikelić_2021} considered specifications 
other than plain linear constraints, namely those expressed by bit-vector logic (BV); see 
Section~\ref{sec:prelims} for details.

\paragraph{$\textsc{Reach}_\mathcal{F}(\text{BV})$:}
given a quantised FNN $N|_{\arith}$ with input dimensionality $d$ for some $\arith=\arith_{n_1,n_2} \in \mathcal{F}$, and 
two BV formulas $\varphi_1,\varphi_2$ both with bit-length $\ell = n_1+n_2$,
decide whether there is $\bs{x} \in \arith^d$ that satisfies $\varphi_1$
such that $N|_{\arith}(\bar{x})$ satisfies $\varphi_2$.

While $\textsc{Reach}_\mathcal{F}(\text{LP})$ and $\textsc{Reach}_\mathcal{F}(\text{BV})$ realistically assume restrictions imposed by
finite-width arithmetic, they permit further generalisation: quantisation often occurs post-training, cf.\ 
\cite{quantizationSurveyGholami2021}, implying that the FNN is initially trained and delivered using more
exacting representations. In actual implementations, it is quantised to different small-width
arithmetics, meaning that the size of the FNN and the size of the finite-width arithmetic are
independent.

\paragraph{$\textsc{Reach}(\mathcal{F}, \text{LP})$:}
given a finite-width arithmetic
$\arith \in \mathcal{F}$, a rational FNN $N$ with input dimensionality $d$, and two rational linear programs $L_1,L_2$, decide whether there exists
$\bs{x} \in \arith^d$ that satisfies $L_1|_{\arith}$ such that $N|_{\arith}(\bs{x})$ satisfies $L_2|_{\arith}$.

\paragraph{$\textsc{Reach}(\mathcal{F}, \text{BV})$:}
given a finite-width arithmetic
$\arith \in \mathcal{F}$, a rational FNN $N$ with input dimensionality $d$, and two BV formulas $\varphi_1,\varphi_2$ with bit-length $\ell = n_1+n_2$, 
decide whether there exists $\bs{x} \in \arith^d$ that satisfies $\varphi_1$ such that $N|_{\arith}(\bs{x})$ 
satisfies $\varphi_2$.

We remark that post-training quantisation techniques, as assumed in $\textsc{Reach}(\mathcal{F}, \text{LP})$ and 
$\textsc{Reach}(\mathcal{F}, \text{BV})$, typically involve more than merely rounding numerical parameters using standard
rounding methods specified by some finite-width arithmetic $\arith$, like rescaling weights prior to quantisation.
Nonetheless, we note that such variations do not affect the aforementioned problems
because they examine the safety of the resulting quantised FNN, independent of its relation to
the unquantised version.

%% file: sections/lp_problems.tex

$\textsc{Reach}_\mathbb{Q}(\text{LP})$ is the problem considered most often in the context of reachability
w.r.t.\ linear constraints. It assumes FNN to work over full-precision, rational arithmetic 
and input- and output-specifications given by sets of linear constraints, i.e.\ linear programs.
\begin{theorem}[\cite{10.1007/978-3-030-89716-1_10}]
    \label{sec:lpproblems;th:standard}
    $\textsc{Reach}_\mathbb{Q}(\text{LP})$ is \NP-complete.
\end{theorem}

We briefly recapitulate the proof of Theorem~\ref{sec:lpproblems;th:standard} to reuse certain arguments for subsequent results. 
\NP-membership is shown using the certificate-based understanding of \NP: guess a witness $w$ of polynomial size in
$\size{}{N} + \size{}{L_1} + \size{}{L_2}$ which is the size of an instance $(N, L_1, L_2)$ of $\textsc{Reach}_\mathbb{Q}(\text{LP})$,
and then verify in polynomial time relative to $\size{}{w} + \size{}{N} + \size{}{L_1} + \size{}{L_2}$ whether $(N, L_1, L_2)$
is a valid instance. Here, the key insight is that for each input $\bs{x}$ of $N$,
there is a fixed \emph{activation pattern} of $N$, meaning a tuple $w \in \{0,1\}^n$, where $n$ is the number of nodes in $N$, 
such that $w[i] = 1$ if and only if the input of the $i$th node in $N$ is greater than or equal to $0$.
Clearly, such an activation pattern $w$ is of size polynomial in the size of $N$.
Having this, the argument for \NP-membership 
is as follows: given FNN $N$, LP $L_1$ and $L_2$, 
(\textsc{i}) guess an activation pattern $w$ of $N$.
(\textsc{ii}) Fix all ReLU nodes of $N$ by replacing the activation of the $i$th node by
    the identity if $w[i]=1$, and by the constant $0$-function if $w[i]=0$.
(\textsc{iii}) Construct a single LP $L$ by combining $L_1$, $L_2$, and $N$, which is possible since $N$
    computes a linear function now, and add constraints ensuring that any solution respects $w$, and
(\textsc{iv}) solve $L$.
Since $L$ is essentially just the combination of $N$, $L_1$, and $L_2$, and LPs can be solved in polynomial time,
the procedure above runs in nondeterministic polynomial time. 

\begin{figure}[t]
\centering
\setlength{\abovedisplayskip}{4pt}
\setlength{\belowdisplayskip}{4pt}
\[
\begin{array}{rl}
\varphi = \bigwedge_{i=1}^n C_i
&: \relu(\sum_{i=1}^n x_{C_i} - (n-1))
\\[2pt]
C_i = \bigvee\limits_{j=1}^3 \ell_{i,j}
&: \relu(\relu(\sum\limits_{j=1}^{3} x_{\ell_{i,j}}) - \relu(\sum\limits_{j=1}^{3} x_{\ell_{i,j}} {-} 1)) 
\\[2pt]
\ell_{i,j} = X_{i,j} &: \relu(x_{i,j}) 
\\[2pt]
\ell_{i,j} =\neg X_{i,j} &:  \relu(1-x_{i,j})
\\[2pt]
x \in \{0,1\} &: \relu(\relu(\frac{1}{2}-x)+\relu(\frac{1}{2}-x)-\frac{1}{2})
\end{array}
\]
\caption{FNN gadgets, one for each component of a propositional 3CNF formula $\varphi$, that output $1$ if the corresponding 
subformula is satisfied and $0$ otherwise. Additionally, it shows a gadget that for $x \in [0;1]$ outputs $0$ if $x=0$ or 
$x=1$, and some value greater than $0$ otherwise.}
\label{sec:lpproblems;fig:3cnf-gadgets}
\end{figure}

\NP-hardness of $\textsc{Reach}_\mathbb{Q}(\text{LP})$ is shown via a reduction from \textsc{3Sat}, the satisfiability
problem of propositional Boolean formulae in three-conjunctive normal form (3CNF). Given a 3CNF
formula $\varphi = \bigwedge_{i=1}^n C_i$ with $n \in \nats$ where each $C_i = (\ell_{i,1} \lor \ell_{i,2} \lor \ell_{i,3})$ is a clause with literals $\ell_{i,j}$ of the form $X$ or $\neg X$ for propositional variables $X \in \mathcal{X}$, we construct
an FNN $N_\varphi$ that accepts precisely those inputs $\bs{x}$ that encode a satisfying assignment $I : \mathcal{X} \rightarrow \{0,1\}$ 
for the propositional variables used in $\varphi$. 
The construction of $N_\varphi$ involves combining specific gadgets as presented in
Fig.~\ref{sec:lpproblems;fig:3cnf-gadgets}. We have one gadget for each kind of
subformula in $\varphi$, mimicking their semantics in a straightforward manner, and combine them as 
determined by the structure of $\varphi$. To ensure
that $N_\varphi$ only accepts inputs corresponding to an encoded variable assignment,
which is essential for a correct reduction from $3\textsc{Sat}$, the FNN $N_\varphi$ also
includes an ``$x \in \{0,1\}$'' gadget (see Fig.~\ref{sec:lpproblems;fig:3cnf-gadgets})
for each input dimension. Additionally, we use $L_1$ to ensure that each input is within $[0,1]$, which is necessary 
for the idea underlying the ``$x \in \{0,1\}$''-gadgets, as well as
$L_2$ to ensure that the output of the final ``$\bigwedge_{i=1}^n C_i$'' gadget in $N_\varphi$ is exactly $1$ and
the output of all 
``$x \in \{0,1\}$'' gadgets is exactly $0$. For details, we
refer to \cite{10.1007/978-3-030-89716-1_10}.

Turning to quantised settings, we find that while specific arguments for Theorem~\ref{sec:lpproblems;th:standard} 
require additional support, the overall complexity of verifying reachability in the context of linear constraints 
remains unchanged. Remember that $\textsc{Reach}_\mathcal{F}(\text{LP})$ takes as input FNN and LP that are 
already quantised.
\begin{theorem}
    \label{sec:lpproblems;th:unary}
    $\textsc{Reach}_\mathcal{F}(\text{LP})$ is \NP-complete.
\end{theorem}
\begin{proof}
    Membership in \NP relies on the certifier-based
    understanding of \NP and is straightforward. Let $(N|_\arith, L_1|_\arith, L_2|_\arith)$ with
    $\arith = \arith_{n_1, n_2}$. 
    (\textsc{i}) Guess an input $\bs{x} \in \arith$,
    (\textsc{ii}) check whether $L_1|_\arith$ is satisfied,
    (\textsc{iii}) compute $N|_{\arith}(\bs{x})$, and
    (\textsc{iv}) check whether $N|_{\arith}(\bs{x})$ satisfies $L_2|_\arith$. 
    Clearly, $\size{}{\bs{x}}$ is linear in $n_1 + n_2$, which is polynomial in $\size{}{N}$
    since $N$ contains numerical parameters represented in $\arith$. Furthermore, steps (\textsc{ii}) to (\textsc{iv})
    can be done in polynomial time with respect to
    $\size{}{L_1} + \size{}{N} + \size{}{L_2}$ since all numerical parameters are represented in $\arith$. 
    Note that we assume constant time for basic arithmetic operations.
    The correctness is implied by Theorem~\ref{thm:lp-solution-size}, which states that feasible LP have 
    solutions of polynomial size. The reasoning used for Theorem~\ref{sec:lpproblems;th:standard}, as  
    sketched above, carries over to parts (\textsc{i})--(\textsc{iv}).

    Regarding \NP-hardness, 
    we refer to the
    same argument used in Theorem~\ref{sec:lpproblems;th:standard}, namely, the reduction from $3\textsc{Sat}$ translating 
    a 3CNF formula $\varphi$ into an FNN $N_\varphi$.
    The key insight is that the same reduction described above works here because the gadgets shown in Figure~\ref{sec:lpproblems;fig:3cnf-gadgets}
    use only low-precision parameters, specifically $1$, $0$, and $\frac{1}{2}$, which are exactly representable in all non-trivial finite-width arithmetics $\arith$.
    Moreover, the sums in the ``$\bigwedge_{i=1}^n C_i$'' and ``$\bigvee_{j=1}^3 \ell_{i,j}$'' gadgets are
    at most $n$. Thus, it is sufficient to assume that $N_\varphi|_\arith$ to be quantised by any $\arith$ with 
    at least $\lceil \log n\rceil$ bits to represent integer parts.
\end{proof}

The other quantised problem we consider in this section is $\textsc{Reach}(\mathcal{F}, \text{LP})$. In contrast to
$\textsc{Reach}_\mathcal{F}(\text{LP})$, the given FNN contains full-precision rational parameters, and the
quantisation to some finite-width arithmetic $\arith_{n_1,n_2}$ is not assumed a priori, making
$\arith_{n_1,n_2}$ part of the input. At first glance, this may lead one to thinking that this problem 
must be more difficult, since we have to deal with numbers of bit size $n_1+n_2$, and the parameters
$n_1,n_2$ are given as binary numbers. However, it is easily argued that this is not the case.
\begin{theorem} 
    \label{sec:lpproblems;th:binary}
    $\textsc{Reach}(\mathcal{F}, \text{LP})$ is \NP-complete.
\end{theorem}
\begin{proof}
    $\NP$-membership is shown as in Theorem~\ref{sec:lpproblems;th:unary}:  
    given $(\arith_{n_1,n_2}, N, L_1, L_2)$, we quantise $N$, $L_1$, and $L_2$ using $\arith_{n_1,n_2}$, 
    guess an input $\bs{x} \in \arith_{n_1,n_2}$ and check whether it is a witness for the validity of $(N, L_1, L_2)$.
    Note that this assumes that the quantisation method used for $N$, $L_1$ and $L_2$ works in polynomial time.
    Given that, we need to argue why the size of $\bs{x}$ is polynomial in the size of $(\arith_{n_1,n_2}, N, L_1, L_2)$. 
    We make a case distinction. First, assume that $\size{}{n_1} + \size{}{n_2}$ is larger than the maximum size of any 
    rational parameter of $N$, $L_1$ or $L_2$. Then, guessing $\bs{x}$, which is polynomial in $\size{}{n_1} + \size{}{n_2}$, 
    is sufficient since this also subsumes solutions for the LP resulting from combining $(N, L_1, L_2)$, utilising the 
    activation pattern of $N$ for $\bs{x}$, according to Theorem~\ref{thm:lp-solution-size}.
    Second, assume that $\size{}{n_1} + \size{}{n_2}$ is smaller than the maximum size of any 
    rational parameter of $N$, $L_1$ or $L_2$. Then, guessing $\bs{x}$, which is polynomial in 
    $\size{}{N} + \size{}{L_1} + \size{}{L_2}$ is also sufficient, again using Theorem~\ref{thm:lp-solution-size}.
    In combination, it is sufficient to take $\bs{x}$ which is polynomial in the size of $(\arith_{n_1,n_2}, N, L_1, L_2)$.
    
    \NP-hardness is shown similarly to Theorems~\ref{sec:lpproblems;th:standard} and~\ref{sec:lpproblems;th:unary}
    by establishing a reduction from $3\textsc{Sat}$. In this case, we introduce
    a finite-width arithmetic $\arith_{n_1,n_2}$ with a sufficient number of bits. We observe that $n_1+n_2$
    is polynomial in the size of the 3CNF formula $\varphi$.
\end{proof}

%% file: sections/bv_problems.tex

As seen in Section \ref{sec:lpproblems}, using linear constraints always guarantees solutions of polynomial size 
with respect to the size of the FNN and specification LP. \cite{Henzinger_Lechner_Žikelić_2021} investigated the 
verification of quantised FNNs with respect to 
bit-vector logics, cf.\ \cite{DBLP:journals/mst/KovasznaiFB16}, instead. This is motivated by the fact that BV allow 
non-linear constraints to be specified and single bits of the FNN's in- and output to be addressed. This can avoid 
rounding and overflow problems in contrast to the standard specifications via linear constraints, which crucially 
depend on the type of arithmetic and bit-width used when evaluating an FNN.

First, for already quantised FNN, these non-linear constraints do not change anything complexity-wise.

\begin{theorem}
    \label{sec:bvproblems;th:unary}
    $\textsc{Reach}_\mathcal{F}(\text{BV})$ is \NP-complete.
\end{theorem}

\begin{proof}
    Membership in \NP is shown as in Theorem~\ref{sec:lpproblems;th:unary}. Let $(N|_\arith, \varphi_1, \varphi_2)$ with
    $\arith = \arith_{n_1, n_2}$. 
    (\textsc{i}) Guess an input $\bs{x} \in \arith$,
    (\textsc{ii}) check whether $\bs{x} \models \varphi_1$, 
    (\textsc{iii}) compute $N|_{\arith}(\bs{x})$, and
    (\textsc{iv}) check whether $N|_{\arith}(\bs{x}) \models \varphi_2$. 
    First, $\size{}{\bs{x}}$ is linear in $n_1 + n_2$ which is polynomial in $\size{}{N}$
    since $N$ contains numerical parameters represented in $\arith$. By Prop.~\ref{prop:bv_nlogspace}, 
    model checking for BV formulas can be done in polynomial time.

    Regarding \NP-hardness,
    we use the same argument as in Theorem~\ref{sec:lpproblems;th:unary}. As the lower bound construction does not 
    need any input or output specification, \NP-hardness already holds for BV specifications as well.
\end{proof}

Analogously to Section \ref{sec:lpproblems}, problem $\textsc{Reach}(\mathcal{F}, \text{BV})$ assumes the FNN given with full-precision parameters together with some finite-width arithmetic as part of the input. This problem was already considered in \cite{Henzinger_Lechner_Žikelić_2021}, where only a lower bound was shown.

\begin{theorem}[\cite{Henzinger_Lechner_Žikelić_2021}]
    \label{thm:bvproblems:pspace_hardness}
    $\textsc{Reach}(\mathcal{F}, \text{BV})$ is \PSPACE-hard.
\end{theorem}

The remainder of this section is dedicated to establishing upper bounds. While the lower bound in \cite{Henzinger_Lechner_Žikelić_2021} is independent of the arithmetic, the upper bounds depend on them. 
We therefore distinguish two cases and establish the following.

\begin{theorem}
    \label{thm:bvproblems_pspace}
    $\textsc{Reach}(\mathcal{F}, \text{BV})$ is \PSPACE-complete for fixed-point and floating-point arithmetic with constant exponent width.
\end{theorem}

We show this by reducing $\textsc{Reach}(\mathcal{F}, \text{BV})$ to the problem of checking 
non-emptiness for non-deterministic finite automata (NFA), represented in a succinct way. Using automata to recognise 
relations induced by linear arithmetics (cf.\ \cite{DBLP:conf/tacas/WolperB00} and 
\cite[Chp.~14]{esparzaAutomataTheoryAlgorithmic2023a}) or for recognising the input-output behaviour of rational 
FNN in particular \cite{10.1007/978-3-031-66159-4_19} was already explored in previous work. We show that these 
results can be extended to quantised FNNs as well and use it for a complexity analysis.

Unlike explicit automata defined by state transition tables, a \emph{succinct NFA} describes a machine 
with an exponentially large state space via a compact, implicit representation.
Formally, an NFA with state space of size $2^n$ is succinctly represented if each state has a polynomial 
description size in $n$ and there are two polynomial-time algorithms: (1) \textsc{Trans} takes two state 
descriptions $q, q'$ and a symbol $\sigma$ and decides whether there is a transition from $q$ to $q'$ with 
symbol $\sigma$. \textsc{Final} takes a state description $q$ and decides whether $q$ is a final state. 
The problem of deciding non-emptiness for such NFAs is equivalent to the reachability problem in succinctly 
represented graphs, which was explored in previous work.

\begin{proposition}[\cite{DBLP:journals/iandc/GalperinW83}]
\label{prop:succinctNFAPSPACE}
    The non-emptiness problem for succinctly represented NFA is in \PSPACE.
\end{proposition}

Succinctly represented NFAs are closed under intersection and union. The state of the product/union automaton 
is the tuple of the states of the component automata. The size of the combined state description is the sum 
of the individual sizes, which remains polynomial. Functions \textsc{Trans} checks for all components to be true. 
Function \textsc{Final} checks that both states are final (for intersection) or at least one is (for union).

The input-output behaviour of rational FNNs evaluated over a given fixed-width arithmetic can be captured by 
succinct NFA. We start with the case of fixed-point arithmetic, encoding fixed-point numbers over the alphabet 
$\Sigma=\{0,1\}$ with least-significant bit first and in two-complement representation. I.e.\ 
the word
$b_f b_{f-1} \cdots b_1 a_0 a_1 \cdots a_{b-1} \in \{0,1\}^{b+f}$
is the encoding of the number $x$ with
\[
x = \sum_{i=1}^f b_i \cdot 2^{-i} + \sum_{i=0}^{b-2} a_i \cdot 2^i + a_{b-1} \cdot -2^{b-1}\ .
\]
Now fix an FNN $N$ with parameters $\theta$, input and output dimensions $m$ and $n$, and a fixed-point arithmetic $\fixarith_{b, f}$ with bit-width $k=b+f$. Its input-output relation is $R_{N,\fixarith_{b, f}} = \{(x, y) \in\{0,1\}^{mk} \times \{0,1\}^{nk} \mid N|_{\fixarith_{b, f}} (x) = y\}$.

\begin{lemma}
    There exists a succinct NFA $M_{N,\fixarith_{b, f}}$ with description size in $\text{poly}(|N|+\log(b+f))$ accepting $R_{N,\fixarith_{b, f}}$.
\end{lemma}
\begin{proof}
    We construct $M_{N,\fixarith_{b, f}}$ to process the input-output relation bitwise. A state $q_t=(t, \bs{c}, \bs{r}, \bs{c}_{\text{init}}, \bs{w})$ is a tuple, where $t \in \{0, \cdots, k\}$ is a bit counter, $\bs{c}\in\mathbb{Z}^{|\theta|}$ represents a vector of carry values for all linear operations, $\bs{r}\in\mathbb{Q}^{|\theta|}$ is a fractional accumulator for rounding and $\bs{w} = \{0,1\}^{|\theta|}$ is the activation pattern. 
    At $t=0$, the automaton non-deterministically guesses the activation pattern $\bs{w}$ and the vector $\bs{c}_{\text{init}}$, which represents the carry bit that would be generated by rounding the fractional result of the scalar products. The accumulator $\bs{r}$ is initialised to 0. The accumulator aggregates computations shifted below the LSB, in order to verify the rounding decisions.
    
Note that for a single neuron with $i$ inputs and maximum weight $\lambda$, the carry is bounded 
by $\mathcal{O}(i\cdot \lambda)$ and especially independent of the bit-width $k$. Additionally, the fractional 
accumulator is always bounded by the size of the weights, which are multiplied. Thus, each state needs space 
$\mathcal{O}(\log k + |N|)$ which is polynomial in the input size. Note that weights are not stored in the state. 
Instead, we use a function $\text{GetBit}(w, j)$ that computes the $j$-th bit 
of a rational weight $w=p/q$ in polynomial time (see Lemma~15 in Appendix).

At step $t$, $M_{N,\fixarith_{b, f}}$ reads the $t$-th bits of input $x$ and output $y$. It updates the carry 
state from $\bs{c}$ to $\bs{c}'$ by using the $t$-th bit of each weight, which can be computed in polynomial time 
(for detailed information on the calculation of carry values consider \cite{DBLP:conf/tacas/WolperB00}). Regarding 
the activation pattern, if $\bs{w}_i=1$ then it checks whether $y_i$ is the correct sum bit or $y_i=0$, when 
$\bs{w}_i=0$, meaning that the corresponding ReLU is inactive.

For $t=k$, a state is accepting when the corresponding sign bit of every neuron's internal sum is consistent 
with the activation pattern. Additionally, the automaton verifies consistency. For ReLU, if the accumulated 
carry indicates that the result should have been negative, but $\bs{w}_j=1$, or positive but $\bs{w}_j=0$, this 
may imply overflow or incorrect guessing. Specifically, if the carry into the sign bit differs from the carry out, 
an overflow occurred. The NFA checks that, if overflow to a positive value occurred, the output $y$ represents 
the maximum representable value, and similarly for a negative overflow. Moreover, it is verified that the 
accumulated tail $\bs{r}$ justifies the initial guess $\bs{c}_{\text{init}}$. For the rounding mode 
round-to-nearest, it checks whether $\bs{r} \geq 0.5 \iff c_{\text{init}}=1$.
\end{proof}

Next, we will consider the case of floating-point arithmetic. The crucial difference here 
is that while addition and multiplication in fixed-point arithmetic are local operations, meaning that the 
computation of the $i$-th bit depends only on the $i$-th bit of the operands and a carry from position $i-1$, 
in floating-point arithmetic addition and multiplication involve an alignment step. When two numbers 
$x_1=m_1\cdot 2^{e_1}$ and $x_2=m_2\cdot 2^{e_2}$ have different exponents, the mantissa of the smaller number 
must be right-shifted by $\delta=|e_1-e_2|$. Therefore, a finite automaton reading the numbers bitwise would 
need to buffer $\mathcal{O}(\delta)$ bits to align the operands. As $k$ is given in binary, this would lead to 
a description of the resulting NFA of exponential size. It remains open whether there is another way to obtain 
a \PSPACE upper bound for the general floating-point arithmetic or whether the lower bound of 
\cite{Henzinger_Lechner_Žikelić_2021} can be lifted. However, to obtain a \PSPACE upper bound when using 
floating-point arithmetic, we can restrict to formats where the exponent width $e$ is fixed, while the 
precision of the mantissa is variable. This is reasonable for hardware designs where dynamic range is static 
but precision varies (for example bf16 vs. fp32).

We encode floating-point numbers with bounded exponent width $E$ and mantissa width $k$ as a word
$se_{0} \cdots e_{E-1} m_1 \cdots m_{k} \in \{0,1\}^{E+k+1}$
encoding the number
\[
x= (-1)^s \cdot 2^{\sum_{i=0}^{E-1} e_i \cdot 2^i} \cdot (1+\sum_{i=1}^k m_i \cdot 2^{-i})\ .
\]

\begin{table*}[t]
    \vspace{1em}
    \centering
        \scalebox{0.75}{%
        \setlength{\tabcolsep}{10pt}
\renewcommand{\arraystretch}{1.5}
        \begin{tabular}{lcc}
        \toprule
        & \textbf{Linear specifications} & \textbf{Bit-vector specifications} \\
        \midrule
        $\textsc{Reach}_\mathbb{Q}$ & \NP-complete (\cite{10.1007/978-3-030-89716-1_10}) & \\ 
        \midrule
        $\textsc{Reach}_\mathcal{F}$  & \NP-complete (Thm.~\ref{sec:lpproblems;th:unary}) & \NP-complete (Thm.~\ref{sec:bvproblems;th:unary})\\
        \midrule
        $\textsc{Reach}(\fixarith, \_)$ & \NP-complete (Thm.~\ref{sec:lpproblems;th:binary}) & \PSPACE-complete  \\
        & & (Thm.~\ref{thm:bvproblems_pspace}, \cite{Henzinger_Lechner_Žikelić_2021})\\
        $\textsc{Reach}(\floatarith_E, \_)$ & \NP-complete (Thm.~\ref{sec:lpproblems;th:binary}) & \PSPACE-complete \\
        & & (Thm.~\ref{thm:bvproblems_pspace}, \cite{Henzinger_Lechner_Žikelić_2021})\\
        $\textsc{Reach}(\floatarith, \_)$ & \NP-complete (Thm.~\ref{sec:lpproblems;th:binary}) & \PSPACE-hard, $\in \NEXPTIME$\\
        & & (Prop.~\ref{thm:bvproblems_nexp}, \cite{Henzinger_Lechner_Žikelić_2021})\\
        \bottomrule
        \end{tabular}
        }
    \caption{Overview of the results established in this paper.}
    \label{fig:results}
\end{table*}

\begin{lemma}
    There exists a succinct NFA $M_{N,\floatarith_{k, E}}$ with description size in $\text{poly}(|N|+\log(k))$ accepting $R_{N,\floatarith{_k, E}}$.
\end{lemma}

\begin{proof}
A state $q_t$ in $M_{N,\floatarith_{k, E}}$ is a tuple $(t, \bs{E}, \bs{S}, \bs{c}, \bs{r}, \bs{c}_{\text{init}}, \bs{w})$, where
$(t, \bs{c}, \bs{r}, \bs{c}_{\text{init}}, \bs{w})$ is as in the fixed-point case. Additionally, $\bs{E}\in \{0, \cdots, 2^E-1\}^{|\theta|}$
stores the exponent values for every neuron and $\bs{S}\in (\{0,1\}^E)^{|\theta|}$ stores the last $E$ mantissa
bits of every neuron's activation value. Since $E$ is fixed, $\bs{E}$ and $\bs{S}$ only take polynomial space,
and each state needs space $\mathcal{O}(\log k + |N|)$, which is polynomial in the input size.

The construction relies on $E$ being constant in three ways. \emph{Alignment}: since all exponents lie in $\{0,\ldots,2^E-1\}$, the maximum alignment shift is $2^E-1$, a constant; the mantissa bits that fall into this window are exactly those stored in $\bs{S}_j$, so aligned operands can be recovered from the state alone. \emph{Normalisation}: after addition or multiplication of normalised values the leading bit shifts by a constant amount bounded by $E$, so normalisation is a constant-size state update. \emph{Rounding}: the sub-LSB contributions are tracked by $\bs{r}$ as in the fixed-point case; the rounding decision depends only on $\bs{r}$ and the LSB, so the fixed-point rounding argument carries over directly. The $t$-th bit of every rational weight can be extracted in polynomial time (see Lemma~16 in Appendix), so $\textsc{Trans}$ and $\textsc{Final}$ both run in polynomial time.
\end{proof}

\begin{lemma}
Let $\varphi$ a BV formula over variables $x_1, \ldots, x_k$ and bit-width $b$ (given in binary). There 
exists a succinct NFA $M_\varphi$ with description size in $\text{poly}(|\varphi|+\log(b))$, such that $M_\varphi$ accepts the set of models of $\varphi$.
\end{lemma}
\begin{proof}
Assume $\varphi$ to be in negation normal form, i.e.\ 
it consists of disjunction and conjunctions of atomic formulas of the form $(t_1=t_2)$ or $\neg(t_1 = t_2)$. 
As succinct NFA are closed under union and intersection, it is sufficient to construct succinct NFA for each 
atomic subformula. 

Let $\psi$ be such a term. A state is a tuple $q = (j, f_{eq})$, where $j \in \{0, \dots, b\}$ is a bit-counter 
and $f_{eq} \in \{0, 1\}$ is a flag tracking whether the evaluated terms equal up to position $j$. The 
description size is in $\mathcal{O}(b)$, which is polynomial. The start state is $(0, 1)$. Let $\sigma \in \{0,1\}^{|V|}$ be the input symbol representing the $j$-th bit-slice of the variables. We define a recursive evaluation function $\text{Eval}(t, j, \sigma)$ for a term $t$: If $t = x \in V$ then $\text{Eval}(x, j, \sigma) = \sigma(x)$, if $t = c$ then $\text{Eval}(c, j, \sigma) = c_j$ (the $j$-th bit of $c$), if $t = (t_a \odot t_b)$ for $\odot \in \{\&, |, \oplus\}$ then $\text{Eval}(t, j, \sigma) = \text{Eval}(t_a, j, \sigma) \odot \text{Eval}(t_b, j, \sigma)$ and if $t = {\sim}t_a$ then $\text{Eval}(t, j, \sigma) = 1 - \text{Eval}(t_a, j, \sigma)$.

Since the depth of the term tree is bounded by $|\psi|$, $\text{Eval}$ runs in polynomial time.

Algorithm $\textsc{Trans}((j, f_{eq}), \sigma, (j', f'_{eq}))$ returns true iff $j' = j + 1$ and $f'_{eq} = 1$ if $ f_{eq} = 1 \text{ and } \text{Eval}(t_1, j, \sigma) = \text{Eval}(t_2, j, \sigma)$ or $f'_{eq} =0$ otherwise.
This transition logic ensures that $f_{eq}$ starts at 1 and acts as a ``sink'': once a mismatch is detected at 
position $j$, $f_{eq}$ becomes and remains 0 for all subsequent steps.

Algorithm $\textsc{Final}((j, f_{eq}))$ returns true iff $j = b$ and the flag matches the formula type: if $\psi$ 
is $(t_1 = t_2)$ then accept iff $f_{eq} = 1$ (all bits matched). If $\psi$ is $\neg(t_1 = t_2)$ then accept 
iff $f_{eq} = 0$ (at least one bit mismatch occurred).
\end{proof}

We are now ready to prove Theorem~\ref{thm:bvproblems_pspace}.

\begin{proof}[Proof (of Thm.~\ref{thm:bvproblems_pspace}).]
    We aim to check whether there is an $\bs{x}$ such that $(\bs{x}, N|_{\arith}) \models \varphi_1 \land \varphi_2$. This corresponds to the non-emptiness of the intersection of the FNN automaton and the specification automaton for
$L(M) = L(M_{N, \arith}) \cap L(M_{\varphi_1 \land \varphi_2})$.
Succinct NFA are closed under intersection. We can therefore construct the succinct product NFA and check for 
emptiness using the \PSPACE algorithm from Proposition~\ref{prop:succinctNFAPSPACE}. \PSPACE-hardness follows from Theorem~\ref{thm:bvproblems:pspace_hardness}.
\end{proof}

For general floating-point arithmetic, we can derive an upper bound from Theorem~\ref{sec:bvproblems;th:unary}.

\begin{proposition}
    \label{thm:bvproblems_nexp}
    $\textsc{Reach}(\mathcal{F}, \text{BV})$ is in \NEXPTIME for general floating-point arithmetic.
\end{proposition}

\begin{proof}
    This is a direct consequence of Theorem~\ref{sec:bvproblems;th:unary}. Due to the binary encoding of the bit-width $b$ we can guess a witness of exponential length. Analogously to Theorem~\ref{sec:bvproblems;th:unary}, we can check the witness in polynomial time.
\end{proof}

These results complement the recently established results from \cite{Henzinger_Lechner_Žikelić_2021} with upper bounds. 

%% file: sections/outlook.tex

We established the complexity landscape of verifying quantised FNN (see Table \ref{fig:results} for an overview), showing that the complexity in most cases 
matches the complexity of verification in full rational arithmetic. Current NN verification tools operate under idealised real-valued 
assumptions, while actual implementations use finite-width arithmetic. As shown in \cite{DBLP:conf/sas/JiaR21}, this mismatch
can invalidate formal guarantees. Our results show that, for the standard case of linear specifications,
verification of quantised networks remains \NP. Hence,
developing sound verification tools for quantised FNNs is not only necessary for
correctness, but feasible without any asymptotic overhead compared to existing real-valued solvers.
Moreover, we complemented the previous \PSPACE lower bound for verifying dynamically 
quantised FNN with bit-vector specifications by giving matching 
upper bounds in the cases of fixed-point and floating-point arithmetic with bounded exponent width. We 
identified that in case of general floating-point arithmetic it seems like the problem could be even harder 
than \PSPACE, due to the fact that the buffer size needed for alignment prevents the \PSPACE upper bound. It 
remains open whether this need for alignment can be exploited to show a matching \NEXPTIME lower bound for 
this case. Additionally, investigating the practical implementation of our succinct automata approach could 
offer a promising direction for creating solvers that accurately verify networks in quantised settings.

%% file: sections/appendix.tex
\begin{lemma}
    \label{lem:app:bit_extract_fixed_point}
    Let $w=\frac{p}{q}$ be a rational number given by binary encoded integers $p,q$ and $t$ be a bit-index given in binary. There exists a polynomial time algorithm $\text{GetBit}_{\fixarith}(w, t)$ that computes the $t$-th bit of binary two's complement representation of $w$ in fixed-point format.
\end{lemma}

\begin{proof}
    The integer part of $w$ is determined by the integer division $A = \lfloor \frac{p}{q} \rfloor$. Since $p$ and $q$ are part of the input, computing $A$ takes polynomial time. The bit at position $2^t$ is simply the $t$-th bit of the integer $A$. If $t$ exceeds the bit-width of $A$, the bit is $0$ (or $1$ for negative numbers).

The $i$-th fractional bit of $\frac{p}{q}$ corresponds to:
\[ b_i = \left\lfloor \frac{p \cdot 2^i}{q} \right\rfloor \pmod 2 \]
Calculating $2^i$ explicitly is impossible since $i$ can be exponential in the input size (up to $2^{|k|}$). However, we only need the result modulo 2. The standard algorithm for extracting the $i$-th fractional bit works as follows:
1. Compute the remainder $r_i = (p \cdot 2^{i-1}) \pmod q$.
   This requires computing $2^{i-1} \pmod q$. Since $i$ is given in binary, we can use modular exponentiation. The time complexity is polynomial in the number of bits of the exponent ($i$) and the modulus ($q$).
2. Once $r_i$ is found, the $i$-th bit is determined by comparing $2 \cdot r_i$ with $q$:
   \[ b_i = \begin{cases} 1 & \text{if } 2 \cdot r_i \ge q \\ 0 & \text{otherwise} \end{cases} \]
\end{proof}

\begin{lemma}
    \label{lem:app:bit_extract_floating_point}
    Let $w=\frac{p}{q}$ be a rational number given by binary encoded integers $p,q$ and $t$ be a bit-index given in binary. There exists a polynomial time algorithm $\text{GetBit}_{\floatarith_{E}}(w, t)$ that computes the $t$-th bit of the binary representation of $w$ in floating-point format with maximum exponent $E$.
\end{lemma}

\begin{proof}
First, we determine the unique integer exponent $V$ such that $2^V \le |\frac{p}{q}| < 2^{V+1}$.
This is equivalent to finding $V$ satisfying:
\[ |p| \cdot 2^{-V-1} < q \le |p| \cdot 2^{-V} \]
Since $p$ and $q$ are given in binary, $V$ can be computed in polynomial time.

If $V \ge 2^E$, the number is too large. We return the bit corresponding to infinity (standard convention: exponent all 1s, mantissa 0). 

We return the bit based on which section $t$ falls into: For the sign bit return $1$ if $w$ is negative, else $0$. If $t$ is an exponent bit we need the bit at index $t-1$ of the integer $V$. Since we computed $V$ in Step 1, we simply return the $(t-1)$-th bit of the binary representation of $V$. If an overflow occurred, return 1 and if an underflow occurred, return 0. Let $j = t - E$ be the index of the mantissa bit $m_j$.
The bit $m_j$ corresponds to the value $2^{-j}$ in the normalized significand $1.m_1 \dots m_k$.
Algebraically, the $j$-th mantissa bit is the bit at position $V - j$ of the original number $|\frac{p}{q}|$.

Therefore, we can use the algorithm for extracting bits for fixed-point numbers from Lemma~\ref{lem:app:bit_extract_fixed_point}: $\text{GetBit}_{\fixarith}(|w|,V-j)$.

All steps—calculating $V$, comparing against $2^E$, and extracting the bit at $V-j$—rely on basic arithmetic and modular exponentiation on binary strings of polynomial length. Thus, the algorithm runs in polynomial time.
\end{proof}